\documentclass[prl,twocolumn,showpacs,preprintnumbers,amsmath,amssymb,superscriptaddress]{revtex4}

\usepackage{graphicx, color, graphpap}
\usepackage{amsmath}
\usepackage{enumitem}
\usepackage{amssymb}
\usepackage{amsthm}
\usepackage{pstricks}
\usepackage{float}
\usepackage{multirow}
\usepackage{hyperref}

\long\def\ca#1\cb{} 

\newcommand{\ket}[1]{|#1\rangle}               
\newcommand{\bra}[1]{\langle #1|}              
\newcommand{\dya}[1]{\ket{#1}\!\bra{#1}}
\newcommand{\dyad}[2]{\ket{#1}\!\bra{#2}}        
\newcommand{\ip}[2]{\langle #1|#2\rangle}      

\newcommand{\EC}{\mathcal{E}}
\newcommand{\FC}{\mathcal{F}}

\newcommand{\HC}{\mathcal{H}}
\newcommand{\IC}{\mathcal{I}}

\newcommand{\XC}{\mathcal{X}}

\newcommand{\Tr}{{\rm Tr}}

\renewcommand{\geq}{\geqslant}
\renewcommand{\leq}{\leqslant}

\newcommand{\eqprop}[2]{\stackrel{\tiny{#1}}{#2}}

\newcommand{\ot}{\otimes}
\newcommand{\ad}{^\dagger}

\newcommand*{\id}{\openone}
\newcommand*{\iso}{\cong}

\newcommand{\al}{\alpha }

\newcommand{\lm}{\lambda }

\newcommand{\sg}{\sigma }

\newtheoremstyle{example}{\topsep}{\topsep}%
{}
{}
{\bfseries}
{.}
{   }
{\thmname{#1}\thmnumber{ #2}}
\theoremstyle{example}

\newtheorem{theorem}{Theorem}
\newtheorem{lemma}[theorem]{Lemma}

\theoremstyle{definition}

\begin{document}

\title{Uncertainty relations from simple entropic properties}

\author{Patrick J. Coles}
\affiliation{Department of Physics, Carnegie Mellon University, Pittsburgh, Pennsylvania 15213, USA}
\author{Roger Colbeck}
\affiliation{Perimeter Institute for Theoretical Physics, 31 Caroline Street North, Waterloo, ON N2L 2Y5, Canada}
\author{Li Yu}
\affiliation{Department of Physics, Carnegie Mellon University, Pittsburgh, Pennsylvania 15213, USA}
\author{Michael Zwolak}
\affiliation{Department of Physics, Oregon State University, Corvallis, OR 97331, USA}
\date{\today}

\begin{abstract}
  Uncertainty relations provide constraints on how well the outcomes
  of incompatible measurements can be predicted, and, as well as being
  fundamental to our understanding of quantum theory, they have
  practical applications such as for cryptography and witnessing
  entanglement.  Here we shed new light on the entropic form of these
  relations, showing that they follow from a few simple entropic
  properties, including the data processing inequality.  We prove
  these relations without relying on the exact expression for the
  entropy, and hence show that a single technique applies to several
  entropic quantities, including the von Neumann entropy, min- and
  max-entropies and the R\'enyi entropies.
\end{abstract}


\maketitle

Uncertainty relations form a central part of our understanding of
quantum mechanics, and give a dramatic illustration of the separation
between quantum and classical physics.  They provide fundamental
constraints on how well the outcomes of various incompatible
measurements can be predicted, as first noted by Heisenberg in the
case of position and momentum
measurements~\cite{Heisenberg}.  This and other early
uncertainty relations~\cite{Robertson,Schrodinger} were formulated
using the standard deviation as the measure of uncertainty.

With the advent of information theory, it became natural to develop
relations using entropies to measure
uncertainty~\cite{BiaMyc75,deutsch,kraus,MaaUff88,EURreview1}.
Furthermore, the most recent versions also account for the possibility
of observers holding additional side-information which they can use to
predict the measurement outcomes~\cite{Hall1,RenesBoileau,BertaEtAl},
and the measurements can be arbitrary POVMs (Positive Operator Valued
Measures)~\cite{TomRen2010,ColesEtAl}, which can be thought of as
projective measurements on a possibly enlarged space (see,
e.g.~\cite{NieChu00}).  When formulated in this way, uncertainty
relations can be applied more directly to problems related to
information processing tasks (data compression, transmission over
noisy channels, etc.), or to cryptography, since the quantities
involved (conditional entropies) have direct operational meanings.

Applications of the uncertainty principle go right back to the first
work on quantum cryptography~\cite{Wiesner}, which discussed a
proposal for quantum money, amongst other things.  However, because
they did not account for the possibility of quantum side information,
the uncertainty relations available at the time could not be directly
applied to prove security against arbitrary adversaries, and served
only an intuitional purpose. Following the discovery of
uncertainty relations that account for the possibility of quantum side
information, there have been many direct applications.  They have been
used, for example, as experimentally efficient entanglement
witnesses~\cite{BertaEtAl,LXXLG,PHCFR}, to provide tight finite-key
rates in quantum key distribution~\cite{TLGR} and to prove security of
certain position-based quantum cryptography
protocols~\cite{KMS,BCFGGOS}.

One way to think about uncertainty relations is in the following
tripartite scenario.  Consider a system, $A$, that will be measured
using one of two measurements, $X$ and $Z$, which can be described in
terms of their POVM elements, $\{X_j\}$ and $\{Z_k\}$ (in this work,
we take these sets to be finite).  If $X$ is measured, an observer
(Bob) holding information $B$ is asked to predict the outcome of this
measurement, while if $Z$ is measured, a second observer (Charlie)
holding $C$ is asked to predict the outcome.  In general, the
information $B$ and $C$ held by the observers may be quantum, and,
most generally, the state before measurement is described by a
tripartite density operator, $\rho_{ABC}$.  Uncertainty relations
provide quantitative limits on the prediction accuracy, often giving a
\textit{trade-off} between Bob's ability to predict $X$ and Charlie's
ability to predict $Z$.

There are many different ways to measure uncertainty, and for much of
this paper, we need not specify precisely which measure we are using.
We use $H_K$ to denote a generic measure of uncertainty, which we call
a $K$-entropy.  $H_K(X|B)$ is then a measure of the uncertainty about
the outcome of measurement $X$ given $B$ and, likewise,
$H_{\widehat{K}}(Z|C)$ is a measure of the uncertainty about the outcome
of measurement $Z$ given $C$, where, for our uncertainty relations, we
require the unspecified entropies, $H_K$ and $H_{\widehat{K}}$, to be
closely related as explained later. A tripartite uncertainty relation then gives a lower bound on
$H_K(X|B)+H_{\widehat{K}}(Z|C)$ which depends on the measurements $X$
and $Z$, and reflects their complementarity.  For example, in the case
where $X$ and $Z$ are composed of commuting projectors, so that
there exist states for which both predictions can be correctly made,
this lower bound will be trivial (i.e.~0).

In this work, we show that such uncertainty relations follow from a
few simple entropic properties.  Among them, the data-processing
inequality forms a central part.  Roughly speaking, this states that
if $B$ provides information about $A$, then processing $B$ cannot
decrease the uncertainty about $A$, which is clearly what one would
expect from an uncertainty measure.

We also obtain relations for the bipartite case where only one
measurement will be made (i.e.\ where we only ask Bob to predict the
outcome of the measurement of $X$).  The state-independent relation we
obtain is trivial if $X$ is projective (then there is always a state
for which $H_K(X|B)=0$), but gives an interesting bound for more
general measurements.  Furthermore, we give an additional relation
that depends on the entropy of the initial state.

More precisely, our main result is that for any entropy $H_K$ that
satisfies a particular set of properties (stated below), the relations
\begin{eqnarray}\label{eqn1}
H_K(X|B)+H_{\widehat{K}}(Z|C)\!&\geq&\!\log\frac{1}{c(X,Z)},\\
\label{eqn2}
H_K(X|B)\!&\geq&\!\log\frac{1}{c(X)},\quad\text{ and}\\
H_K(X|B)\!&\geq&\!\log\frac{1}{c'(X)}+H_K(A|B)\label{eqn3}
\end{eqnarray}
hold for any state $\rho_{ABC}$, where
$c(X,Z)=\max_{jk}\|\sqrt{X_j}\sqrt{Z_k}\|^2_{\infty}$,
$c(X)=c(X,\{\id\})$ and $c'(X)=\max_j\Tr(X_j)$ (the infinity norm of
an operator is its largest singular value)~\footnote{While the base of
the logarithm is conventionally taken to be 2, so that entropies are
measured in \emph{bits}, our results apply for any base, provided the
same one is used throughout.}. In~\eqref{eqn3}, $H_K(A|B)$ is the conditional $K$-entropy of $A$ given $B$, and in \eqref{eqn1}, $H_{\widehat{K}}$ is the entropy dual to $H_K$ in the sense that for any pure state $\rho_{ABC}$, $H_K(A|B)+H_{\widehat{K}}(A|C)=0$.

In particular, our proof applies to the von Neumann entropy, the min-
and max-entropies, and a range of R\'enyi entropies.  For the
tripartite relation, the first two cases were already
known~\cite{BertaEtAl,TomRen2010,ColesEtAl}, while the latter is new,
and for the bipartite relations we extend previous work on this
idea~\cite{KrishnaParth, ColesEtAl, Rastegin2008} to allow for other
entropies or quantum side information.  To emphasize, the main
contribution of the present work is that it provides a unified proof
of these relations.\bigskip

\textit{Entropic Properties}.|As mentioned above, we are interested in the uncertainties of POVM outcomes. A POVM, $X$,
can be specified via a set of operators $\{X_j\}$ that satisfy
$X_j\geq 0$, $\sum_j X_j=\id$.  We also define an associated TPCPM
(Trace Preserving Completely Positive Map), $\XC$, from $\HC_A$ to
$\HC_X$ given by
\begin{equation}\label{eqn4}
\XC:\rho_A\mapsto\sum_j\dyad{j}{j}_X\Tr(X_j\rho_A),
\end{equation}
where $\{\ket{j}\}$ form an orthonormal basis in $\HC_X$.  Thus, for a
state $\rho_{AB}$, we can define the conditional $K$-entropy of $X$
given $B$, denoted $H_K(X|B)$, as the conditional $K$-entropy of the
state $(\XC\ot\IC)(\rho_{AB})$.

A (bipartite) conditional entropy is a map from the set of density
operators on a Hilbert space $\HC_{AB}$ to the real numbers.  In turns
out to be convenient to consider a generalized quantity, $D_K(S||T)$,
which maps two positive semi-definite operators to the real numbers.
Such quantities are often called relative entropies.  We consider
relative $K$-entropies that are constructed such that they generalize
the respective conditional $K$-entropies in the sense that, depending
on the entropy, either $H_K(A|B)=-D_K(\rho_{AB}||\id\ot\rho_B)$, or
$H_K(A|B)=\max_{\sigma_B}[-D_K(\rho_{AB}||\id\ot\sigma_B)]$ where $\sg_B$ is any (normalized) density operator on $\HC_B$.

We now introduce the properties of $D_K$ that allow us to prove our
uncertainty relations:
\begin{enumerate}[label=(\alph{enumi})]

\item \label{a} Decrease under TPCPMs: If $\EC$ is a TPCPM, then
$D_K(\EC(S)||\EC(T))\leq D_K(S||T)$.

\item \label{b} Being unaffected by null subspaces: $D_{K}(S \oplus 0 ||
T\oplus T')=D_{K}(S||T)$, where $\oplus$ denotes direct sum.

\item \label{c} Multiplying the second argument: If $c$ is a positive constant, then $D_K(S||cT)=D_K(S||T)+ \log\frac{1}{c}$.

\item \label{d} Zero for identical states: For any density operator $\rho$, $D_K(\rho ||\rho)=0$.

\end{enumerate}

Property~\ref{a} implies the increase of $H_K(A|B)$ under TPCPMs on
$B$, i.e.\ the data processing inequality|doing operations on $B$
cannot decrease the uncertainty about $A$. It also implies that $D_K$
is invariant under isometries $U$, i.e.,
\begin{equation}
\label{eqn4444}
D_K(USU^{\dagger}||UTU^{\dagger})=D_K(S||T).
\end{equation}
This can be seen by invoking~\ref{a} twice in succession, first with
the TPCPM corresponding to $U$, then with a TPCPM that undoes $U$, establishing that $D_K(S||T)\geq
D_K(USU^{\dagger}||UTU^{\dagger})\geq D_K(S||T)$, and
hence~\eqref{eqn4444}.

The uncertainty relation~\eqref{eqn1} is expressed in terms of the
entropy $H_K$ and its dual $H_{\widehat{K}}$, the latter being defined
by $H_{\widehat{K}}(A|B):=-H_K(A|C)$, where $\rho_{ABC}$ is a
purification of $\rho_{AB}$.  That this is independent of the chosen
purification (and hence that $H_{\widehat{K}}$ is well-defined) is
ensured by the invariance of $H_K(A|B)$ under local isometries (shown
in the Appendix, Lemma~\ref{lem:isom}), and the fact that purifications are unique up to isometries on the purifying system
(see, for example,~\cite{NieChu00}).  This definition also ensures
that $H_{\widehat{K}}(A|B)$ inherits many natural properties of
$H_K(A|B)$, for example, increase under TPCPMs on $B$ and invariance
under local isometries.

We proceed by giving some examples of entropies that fit these
criteria.  The first is the von Neumann entropy, which can be defined
via the von Neumann relative entropy.  For two positive operators, $S$
and $T$, this is given by
$$D(S||T):=\lim_{\xi\rightarrow 0}\frac{1}{\Tr S}(\Tr(S\log S)-\Tr(S\log (T+\xi\id))).$$
Note that if $T$ is invertible, the limit is not needed, and if part of $S$
lies outside the support of $T$ then $D(S||T)=\infty$. For a density
operator $\rho_{AB}$, we can then define the conditional von Neumann
entropy of $A$ given $B$ by $H(A|B):=-D(\rho_{AB}||\id\ot\rho_B)$.
The von Neumann entropy is its own dual, i.e.\ for any pure state
$\rho_{ABC}$, we have $H(A|B)=-H(A|C)$.

A second class of entropies to which our results apply are a range of
R\'enyi entropies~\cite{Renyi, Petz84} (for examples of their
application, see e.g.~\cite{MosonyiEtAl2011}).  For positive
operators, $S$ and $T$, and for $\alpha\in(0,1)\cup (1,2]$, the
R\'enyi relative entropy of order $\alpha$ is defined by
$$D_{\alpha}(S||T):=\lim_{\xi\rightarrow 0}\frac{1}{\alpha-1}\log\Tr(S^\alpha (T+\xi\id)^{1-\alpha}).$$
Furthermore, we define
\begin{eqnarray*}
D_0(S||T)&:=&\lim_{\alpha\rightarrow 0+}D_{\alpha}(S||T)\quad \text{ and}\\
D_1(S||T)&:=&\lim_{\alpha\rightarrow 1}D_{\alpha}(S||T)=D(S||T).
\end{eqnarray*}
Hence, the von Neumann relative entropy can be seen as the special
case $\alpha=1$. The relative entropy $D_{\alpha}$ gives rise to the
conditional R\'enyi entropy
$$H_{\al}(A|B):= -D_{\al}(\rho_{AB} || \id\ot\rho_B ),$$
which satisfies the duality relation that $H_{\al}(A|B)=-H_{2-\al}(A|C)$ for pure
$\rho_{ABC}$~\cite{TomColRen09}.

Furthermore, the min and max relative entropies
\begin{align}
&D_{\min}(S || T):=\log\min\{\lm : S\leq \lm T \} \notag \\
&D_{\max}(S || T):= -2\log\Tr\sqrt{\sqrt{S}T\sqrt{S}}\notag
\end{align}
can be used to define the related conditional
entropies~\cite{RennerThesis05,KonRenSch09}
\begin{align}
&H_{\min}(A|B):=\max_{\sg_B} [-D_{\min}(\rho_{AB} || \id\ot \sg_B )] \notag\\
&H_{\max}(A|B):= \max_{\sg_B} [ -D_{\max}(\rho_{AB} || \id\ot \sg_B )] \notag
\end{align}
which satisfy the duality relation
$H_{\min}(A|B)=-H_{\max}(A|C)$~\cite{KonRenSch09}.
We also consider the entropies
\begin{align*}
&\widehat{H}_\alpha(A|B):= \max_{\sg_B} [ -D_\alpha(\rho_{AB} || \id\ot \sg_B )].
\end{align*}
While in general we do not have alternative expressions for the duals
of the latter entropies, it has been shown~\cite{Berta} that
$\widehat{H}_{\min}(A|B)=-\widehat{H}_0(A|C)$ for pure $\rho_{ABC}$,
where $\widehat{H}_{\min}(A|B):=-D_{\min}(\rho_{AB}||\id\ot\rho_B)$.\bigskip

\textit{Main Results}.|Our main result is that the properties discussed above are sufficient to
establish the following uncertainty relations~\footnote{See the Appendix for proofs and elaboration of our results.}.
\begin{theorem}
\label{thm1}
Let $X=\{X_{j}\}$ and $Z=\{Z_{k}\}$ be arbitrary POVMs on $A$, and
$H_K(A|B)$ be such that either
$H_K(A|B)=-D_K(\rho_{AB}||\id\ot\rho_B)$ or
$H_K(A|B)=\max_{\sigma_B}[-D_K(\rho_{AB}||\id\ot\sigma_B)]$, for all
$\rho_{AB}$, where $D_K$ satisfies Properties~\ref{a}--\ref{c}.  It
follows that for all $\rho_{ABC}$
\begin{equation*}
 H_K(X|B)+ H_{\widehat K}(Z|C)\geq \log\frac{1}{c(X,Z)},
\end{equation*}
where $c(X,Z)=\max_{j,k} \| \sqrt{Z_{k}} \sqrt{X_{j}}\|_\infty^2$.
\end{theorem}
The ideas behind this proof are illustrated below where we give a
proof for the special case where $H_K$ is the von Neumann entropy, and
$X$ and $Z$ are composed of rank-one projectors.

We also have the following single-measurement uncertainty relation.
\begin{lemma}
\label{thm2}
Let $X=\{X_{j}\}$ be an arbitrary POVM on $A$, and suppose that $H_K$
and its related $D_K$ satisfy the conditions given in
Theorem~\ref{thm1}, as well as Property~\ref{d}.  Then, for all
$\rho_{AB}$,
\begin{equation}
\label{eqn6}
 H_K(X|B)\geq \log\frac{1}{c(X)},
\end{equation}
where $c(X):=c(X,\{\id\})=\max_{j}\| X_{j}\|_\infty$.
\end{lemma}
\begin{proof}
  This follows from Theorem~\ref{thm1} by setting $Z=\{\id\}$ and
  using the fact that $H_{\widehat K}(Z|C)=0$ in this case (see
  Lemma~\ref{thmpp} in the Appendix).
\end{proof}

However, there is an alternative single-measurement relation, which
can give a stronger bound than~\eqref{eqn6}.
\begin{lemma}
\label{thm3}
Let $X=\{X_j\}$ be an arbitrary POVM on $A$, and $H_K(A|B)$ be such
that either $H_K(A|B)=-D_K(\rho_{AB}||\id\ot\rho_B)$ or
$H_K(A|B)=\max_{\sg_B}[-D_K(\rho_{AB}||\id\ot\sg_B)]$, for all
$\rho_{AB}$, where $D_K$ satisfies Properties~\ref{a}--\ref{c}.  It
follows that
\begin{equation*}
 H_K(X|B)\geq \log\frac{1}{c'(X)}+H_K(A|B),
\end{equation*}
where $c'(X)=\max_j\Tr(X_j)$.
\end{lemma}

We remark that the bounds in these results can be generalized in
the following way.  Suppose $\Pi$ is a projector on $\HC_A$ whose
support includes the support of $\rho_A$.  The above results hold if
$c(X,Z)$ is replaced by $c(X,Z;\Pi):=\max_{j,k} \| \sqrt{Z_{k}}\Pi
\sqrt{X_{j}}\|_\infty^2$, and if $c'(X)$ is replaced by
$c'(X;\Pi)=\max_j\Tr(X_j\Pi)$. See~\cite{HanTom2011} for further ways
to take advantage of knowledge of the state to derive tighter
uncertainty relations for the von Neumann entropy.

We have shown that, in order to establish that a particular entropy
satisfies these uncertainty relations, it suffices to verify that it
satisfies a few properties. (Recall that for any entropy satisfying
our properties, its dual is automatically well defined; it is not
necessary to have an alternative expression for it in order
for~\eqref{eqn1} to hold.)
\begin{lemma}
\label{thm4}
All examples of relative entropies defined above satisfy Properties~\ref{a} through~\ref{d}.
\end{lemma}
\begin{proof}
Properties~\ref{b},~\ref{c}, and~\ref{d} follow directly from the definitions of these entropies. Property~\ref{a} was discussed in, e.g.,~\cite{NieChu00} for the von Neumann relative entropy, in~\cite{Petz84, TomColRen09} for the R\'enyi relative entropies ($D_0$ being a special case), and in~\cite{RennerThesis05} for the min relative entropy.  For the max relative entropy, it follows because the fidelity is monotonically increasing under TPCPMs~\cite{BCFJS}.
\end{proof}

This implies that the dual entropy pairs $(H,H)$,
$(H_{\al},H_{2-\al})$, $(H_{\min}, H_{\max})$ and
$(\widehat{H}_{\min}, \widehat{H}_0)$ each satisfy Eq.~\eqref{eqn1},
and that the entropies $H$, $H_{\al}$, $H_{\min}$, $H_{\max}$,
$\widehat{H}_{\al}$ and $\widehat{H}_{\min}$ each satisfy
Eqs.~\eqref{eqn2} and~\eqref{eqn3}.\bigskip

\textit{Illustration of the proof technique}.|In order to illustrate how our properties combine to yield uncertainty
relations, we give a proof in the special case of the von Neumann
entropy and where $X=\{\dya{X_j}\}$ and $Z=\{\dya{Z_k}\}$ are
orthonormal bases.  Although more straightforward, this proof features
all of the essential ideas of its generalization.  We note that in
this case $c(X,Z)=\max_{j,k} |\ip{X_{j}}{Z_{k}}|^2$, and the resulting
uncertainty relation,
\begin{equation}
\label{eqn8}
 H(X|B)+ H(Z|C)\geq \log\frac{1}{c(X,Z)},
\end{equation}
is the one conjectured in~\cite{RenesBoileau} and proven
in~\cite{BertaEtAl}.

We first show that all relative $K$-entropies are decreasing under
increases of its second argument.
\begin{lemma}
\label{thm5}
If $D_K(S||T)$ satisfies Properties~\ref{a} and~\ref{b}, then for all
positive operators $S$ and $T$, and for $\tilde{T}\geq T$,
\begin{equation}
\label{eqn999}
D_{K}(S||T)\geq D_{K}(S||\tilde{T}).
\end{equation}
\end{lemma}
\begin{proof}  
  Denote $\HC_{\mu}$ as the Hilbert space on which $S$, $T$ and
  $\tilde{T}$ are defined and introduce $\HC_{\nu}$ as an isomorphic
  Hilbert space.  Let $\{\ket{\mu_j}\}$ and $\{\ket{\nu_j}\}$ be
  orthonormal bases for $\HC_{\mu}$ and $\HC_{\nu}$ and let
  $\HC=\HC_{\mu}\oplus\HC_{\nu}$.  We also introduce a TPCPM acting on operators on $\HC$,
  $\FC:S\mapsto F_1 S F_1\ad+F_2 S F_2\ad$, with
  $F_1=\sum_j\dyad{\mu_j}{\mu_j}$ and $F_2=\sum_j\dyad{\mu_j}{\nu_j}$.
  For $W:=\tilde{T}-T$, we have
\begin{eqnarray*}
D_K(S||T)&\eqprop{\ref{b}}{=}&D_K(S\oplus 0 || T\oplus W)\\
&\eqprop{\ref{a}}{\geq}&D_K(\FC(S\oplus 0) ||
  \FC(T\oplus W))\\
&\eqprop{\ref{b}}{=}& D_K(S \oplus 0 || (T+ W)\oplus 0)= D_K(S ||
  \tilde{T}).
\end{eqnarray*}
\end{proof}

Now, define the isometry $V_X:=\sum_j \ket{j}\ot X_j$ associated with
the $X$ measurement on system $A$, and the state
$\tilde{\rho}_{XABC}:= V_X \rho_{ABC}V_X\ad$.  We proceed to give the
proof for the case of pure $\rho_{ABC}$. The impure case follows by
considering a purification, $\rho_{ABCD}$, and using $H_{}(X|C)\geq
H_{}(X|CD)$ (from Property~\ref{a}). Applying the duality to
$\tilde{\rho}_{XABC}$ gives:
\begin{eqnarray}
H(X|C)&=&-H(X|AB)=D(\tilde{\rho}_{XAB}|| \id \ot \tilde{\rho}_{AB})\nonumber\\
  &\eqprop{\ref{b}}{=}& D(V_X\rho_{AB} V_X\ad || V_X\sum_j X_j\rho_{AB} X_j V_X\ad)\nonumber\\
    &\eqprop{\eqref{eqn4444}}{=}&D(\rho_{AB}||\sum_j X_j\rho_{AB} X_j)\nonumber\\
    &\eqprop{\ref{a}}{\geq}& D(\overline{\rho}_{ZB} || \sum_{j,k} |\ip{X_{j}}{Z_{k}}|^2Z_k\ot \Tr_A\{X_j\rho_{AB}\})\nonumber\\
    &\eqprop{\eqref{eqn999}}{\geq}& D(\overline{\rho}_{ZB} || c(X,Z)\id \ot \rho_{B})\nonumber\\
    &\eqprop{\ref{c}}{=}& \log (1/c(X,Z))+D(\overline{\rho}_{ZB} || \id \ot \rho_{B})\nonumber\\
    &=& \log (1/c(X,Z)) - H_{}(Z|B),\label{eq:hxc}
\end{eqnarray}
where we have used $\overline{\rho}_{ZB}:=\sum_k Z_k\rho_{AB}Z_k$.

We note that our proof technique points to a method for finding states
that satisfy the uncertainty relation~\eqref{eqn8} with equality.  In
the case of pure states $\rho_{ABC}$ and mutually unbiased bases $X$
and $Z$ (for which $|\ip{X_j}{Z_k}|$ is independent of $j,k$), the
only inequality remaining is a single use of Property~\ref{a} (the
fourth line of~\eqref{eq:hxc}). In this case,~\eqref{eqn8} is
satisfied with equality if Property~\ref{a} is saturated, for the
particular TPCPM used in the proof.

For the von Neumann relative entropy,~\ref{a} is satisfied with
equality~\cite{Petz2003,HaydenEtAl04} if and only if there exists a
TPCPM, $\hat \EC$, that undoes the action of $\EC$ on $S$ and $T$,
i.e.\
\begin{equation}
\label{eqn12}
(\hat\EC\circ\EC)(S)=S,\quad (\hat\EC\circ\EC)(T) = T.
\end{equation} 
Hence, states of minimum uncertainty are closely connected to the
\emph{reversibility} of certain quantum operations.  For specific
examples, we refer the reader to~\cite{ColesYuZwo2011}.\bigskip

\textit{Acknowledgements}.|We thank Robert Griffiths for helpful conversations. Research at Carnegie Mellon was supported by the Office of Naval Research and by the National Science Foundation through Grant No.\ PHY-1068331. Research at Perimeter Institute was supported by the Government of Canada through Industry Canada and by the Province of Ontario through the Ministry of Research and Innovation.

\bibliography{EntUncRel}

\section*{APPENDIX}

\subsection{Properties from the main article}

Recall, from the main article, that the following properties of the relative entropy $D_K$ allow us to prove our
uncertainty relations:\\

(a) Decrease under TPCPMs: If $\EC$ is a TPCPM, then
$D_K(\EC(S)||\EC(T))\leq D_K(S||T)$.\\

(b) Being unaffected by null subspaces: $D_{K}(S \oplus 0 ||
T\oplus T')=D_{K}(S||T)$.\\  

(c) Multiplying the second argument: If $c$ is a positive constant, then $D_K(S||cT)=D_K(S||T)+ \log\frac{1}{c}$.\\

(d) Zero for identical states: For any density operator~$\rho$, $D_K(\rho ||\rho)=0$.\\

\ca
\begin{enumerate}[label=(\alph{enumi})]

\item \label{a} Decrease under TPCPMs: If $\EC$ is a TPCPM, then
$D_K(\EC(S)||\EC(T))\leq D_K(S||T)$.

\item \label{b} Being unaffected by null subspaces: $D_{K}(S \oplus 0 ||
T\oplus T')=D_{K}(S||T)$.  

\item \label{c} Multiplying the second argument: If $c$ is a positive constant, then $D_K(S||cT)=D_K(S||T)+ \log\frac{1}{c}$.

\item \label{d} Zero for identical states: For any density operator $\rho$, $D_K(\rho ||\rho)=0$.
\end{enumerate}
\cb

Also, recall the following consequences of these properties. Property~\ref{a} implies that $D_K$ is invariant under isometries $U$, i.e.,
\begin{equation}
\label{eqn4956}
D_K(USU^{\dagger}||UTU^{\dagger})=D_K(S||T),
\end{equation}
and Properties~\ref{a} and \ref{b} together imply that for all positive operators $S$ and $T$, and for $\tilde{T}\geq T$,
\begin{equation}
\label{eqn9}
D_{K}(S||T)\geq D_{K}(S||\tilde{T}).
\end{equation}

\subsection{Additional properties}

Here we show some additional useful relations that follow from our
properties.  
\begin{lemma}\label{thm7}
If $D_K$ satisfies Properties~\ref{a} and~\ref{d}, then for density
  operators $\rho$ and $\sigma$,
\begin{equation}\label{eqn13}
D_K(\rho||\sigma)\geq 0,
\end{equation}
with equality if $\rho=\sigma$.
\end{lemma}
\begin{proof}
  Using Property~\ref{a}, we have $D_K(\rho||\sigma)\geq
  D_K(\rho||\rho)$, where we consider the TPCPM corresponding to
  tracing the original system and replacing it with $\rho$.  The
  result then follows from~\ref{d}.
\end{proof}
We note that \eqref{eqn13} applies only to the case where $\rho$ and $\sg$ are normalized (i.e., unit trace), and the relative entropy $D_K(S||T)$ can be either positive or negative for unnormalized positive operators $S$ and $T$ (and, in particular, in the case $T=\id \ot \sg_B$). 

The next lemma implies that for conditional entropies of the form
$H_K(A|B)=\max_{\sigma_B}[-D_K(\rho_{AB}||\id\ot\sigma_B)]$, one can
restrict the maximization over $\sigma_B$ to states in the support of
$\rho_B$.

\begin{lemma}\label{lem:supp}
For all $\rho_{AB}$, there exists a state $\tilde{\sigma}_B$ that
  lies in the support of $\rho_B$ such that
  $\max_{\sigma_B}[-D_K(\rho_{AB}||\id\ot\sigma_B)]=-D_K(\rho_{AB}||\id\ot\tilde{\sigma}_B)$,
  where $D_K$ satisfies Properties~\ref{a} and~\ref{b}.
\end{lemma}
\begin{proof}
  Let $\Pi$ be the projector onto the support of $\rho_B$, and
  consider the TPCPM
  $\EC:\sg\mapsto\Pi\sg\Pi\oplus(\id-\Pi)\sg(\id-\Pi)$.  For an
  arbitrary density operator $\bar{\sigma}_B$, defining
  $\hat{\sigma}_B:=\frac{1}{\Tr(\Pi\bar{\sigma}_B)}\Pi\bar{\sigma}_B\Pi$,
  we have
\begin{eqnarray*}
  D_K(\rho_{AB}||\id\ot\bar{\sigma}_B)&\eqprop{\ref{a}}{\geq}&D_K(\rho_{AB}||\id\ot\EC(\bar{\sigma}_B))\\
&\eqprop{\ref{b}}{=}&D_K(\rho_{AB}||\id\ot\Pi\bar{\sigma}_B\Pi)\\
&\eqprop{\eqref{eqn9}}{\geq}&D_K(\rho_{AB}||\id\ot\hat{\sigma}_B).
\end{eqnarray*}
In other words, for any $\bar{\sigma}_B$, the density operator
$\hat{\sigma}_B$ is at least as
good for achieving the maximization over $\sigma_B$, which implies the
stated result.
\end{proof}

The next lemma shows that the conditional entropy is invariant under
local isometries (this property is needed to ensure that the dual
entropy is well-defined; see the main text).

\begin{lemma}\label{lem:isom}
 Let $H_K(A|B)$ be such that either
  $H_K(A|B)=-D_K(\rho_{AB}||\id\ot\rho_B)$ or
  $H_K(A|B)=\max_{\sigma_B}[-D_K(\rho_{AB}||\id\ot\sigma_B)]$, for all
  $\rho_{AB}$, where $D_K$ satisfies Properties~\ref{a} and~\ref{b}.
  For arbitrary (local) isometries $V_A:\HC_A\mapsto\HC_{A'}$ and
  $V_B:\HC_B\mapsto\HC_{B'}$, an arbitrary state $\rho_{AB}$ and
  $\rho_{A'B'}:=(V_A\ot V_B)\rho_{AB}( V_A\ad \ot V_B\ad)$, it follows
  that $H_K(A|B)=H_K(A'|B')$.
\end{lemma}
\begin{proof}
  In the case $H_K(A|B)=-D_K(\rho_{AB}||\id\ot\rho_B)$, we have
\begin{align}
H_K(A|B)& \eqprop{\eqref{eqn4956}}{=}-D_K(\rho_{A'B'} || V_A V_A\ad \ot \rho_{B'})\notag\\
& \eqprop{\ref{b}}{=}-D_K(\rho_{A'B'} || I_{A'} \ot \rho_{B'})=H_K(A'|B'),\notag
\end{align}
as required.  Now consider $H_K(A|B)=
\max_{\sigma_B}[-D_K(\rho_{AB}||\id\ot\sigma_B)]$.  In this case, we
have
\begin{align}
-H_K(A|B)& \eqprop{\eqref{eqn4956}}{=} \min_{\sigma_B} D_K(\rho_{A'B'} || V_A V_A\ad \ot V_B\sigma_B V_B\ad)\notag\\
& \eqprop{\ref{b}}{=} \min_{\sigma_B} D_K(\rho_{A'B'} || I_{A'} \ot V_B\sigma_B V_B\ad)\notag\\
& = \min_{\sigma_{B'}} D_K(\rho_{A'B'} || I_{A'} \ot \sigma_{B'}) =-H_K(A'|B'),\notag
\end{align}
where to get to the last line we used Lemma~\ref{lem:supp}.
\end{proof}

\begin{lemma}\label{thmpp}
  Let $H_K(A|B)$ be such that either
  $H_K(A|B)=-D_K(\rho_{AB}||\id\ot\rho_B)$ or
  $H_K(A|B)=\max_{\sigma_B}[-D_K(\rho_{AB}||\id\ot\sigma_B)]$, for all
  $\rho_{AB}$, where $D_K$ satisfies Properties~\ref{a}--\ref{d}.
  Then
\begin{enumerate}[label=(\roman{*})]
\item \label{pt1} $H_K(A|B)=0$ for any state of the form
  $\rho_{AB}=\dya{\psi}_A\ot\rho_B$.
\item \label{pt2} $H_K(A|B)\leq\log d_A$, where $d_A$ is the dimension of $\HC_A$.
\end{enumerate}
\end{lemma}
\begin{proof}
  For~\ref{pt1}, note that for such a state,
\begin{eqnarray*}
H_K(A|B)&=&\max_{\sg_B}[-D_K(\dya{\psi}_A\ot\rho_B||\id\ot\sigma_B)]\\
&\eqprop{\ref{b}, \eqref{eqn4956}}{=}&\max_{\sg_B}[-D_K(\rho_B||\sigma_B)]\eqprop{\eqref{eqn13}}{=}0
\end{eqnarray*}
(the case $H_K(A|B)=-D_K(\rho_{AB}||\id\ot\rho_B)$ follows similarly).

For~\ref{pt2}, note that
\begin{eqnarray*}
  D_K(\rho_{AB}||\id\ot\sg_B)&\eqprop{\ref{c}}{=}&\!\!D_K(\rho_{AB} ||\frac{\id}{d_A}\ot\sg_B)+\log\frac{1}{d_A}\\
&\eqprop{\eqref{eqn13}}{\geq}&\!\!\log\frac{1}{d_A}.
\end{eqnarray*}
The result then follows from the definition of $H_K$.  
\end{proof}
Note that if Part~\ref{pt1} of Lemma~\ref{thmpp} 
holds for $H_K(A|B)$, it also holds for the dual entropy, $H_{\widehat K}(A|B)$.  It is
additionally worth noting that if both $H_K$ and $H_{\widehat K}$
satisfy the conditions of Lemma~\ref{thmpp}, 
then it also follows that $H_K(A|B)\geq-\log d_A$.

\begin{lemma}
  Let $H_K(A|B)$ be such that either
  $H_K(A|B)=-D_K(\rho_{AB}||\id\ot\rho_B)$ or
  $H_K(A|B)=\max_{\sigma_B}[-D_K(\rho_{AB}||\id\ot\sigma_B)]$, for all
  $\rho_{AB}$, where $D_K$ satisfies Properties~\ref{a}--\ref{d}.
  Then, for any POVM, $X$, on $A$, we have $H_K(X|B)\geq 0$ with
  equality if the states $\Tr_A(X_j\rho_{AB})$ have orthogonal
  support.  
\end{lemma}
\begin{proof}
  That $H_K(X|B)\geq 0$ follows from Lemma~2 of the main article.  The equality
  condition follows from Property~\ref{b} and Lemma~\ref{thm7}. 
\end{proof}

\subsection{Proof of Theorem~1}
\begin{proof}
  The proof uses an idea in~\cite{TomRen2010} to consider the
  measurements in terms of two isometries.  We introduce $V_X$ and
  $V_Z$ as two isometries from $\HC_A$ to either
  $\HC_X\ot\HC_{X'}\ot\HC_{A}$ or $\HC_Z\ot\HC_{Z'}\ot\HC_{A}$
  respectively:
\begin{eqnarray}
  V_X&:=&\sum_j\ket{j}_X\ot\ket{j}_{X'}\ot\sqrt{X_j}\label{eqnten}\\
  V_Z&:=&\sum_k\ket{k'}_Z\ot\ket{k'}_{Z'}\ot\sqrt{Z_k},
\end{eqnarray}
where $\{\ket{j}\}$ is an orthonormal basis on $\HC_X\iso\HC_{X'}$ and
similarly $\{\ket{k'}\}$ is an orthonormal basis on
$\HC_Z\iso\HC_{Z'}$.  We consider first the case where $\rho_{ABC}$ is
pure.  Applying the duality to the (pure) state
$\overline{\rho}_{ZZ'ABC}=V_Z\rho_{ABC}V_Z\ad$, gives $H_{\widehat
  K}(Z|C)=-H_K(Z|Z'AB)$.  Then note that
\begin{widetext}
\begin{eqnarray}
  D_{K}(\overline{\rho}_{ZZ'AB}|| \id_Z \ot \sg_{Z'AB})&\eqprop{\ref{a},\ref{b}}{\geq}&D_{K}(\overline{\rho}_{ZZ'AB}|| V_Z V_Z \ad (\id_Z \ot \sg_{Z'AB}) V_Z V_Z \ad)\nonumber\\
  &\eqprop{\eqref{eqn4956}}{=}&D_{K}(\rho_{AB}|| V_Z \ad (\id_Z \ot \sg_{Z'AB}) V_Z)\nonumber\\
  &\eqprop{\eqref{eqn4956}}{=}&D_{K}(\tilde{\rho}_{XX'AB} || V_X V_Z \ad (\id_Z \ot \sg_{Z'AB}) V_Z V_X\ad)\nonumber\\
  &=&D_{K}\big(\tilde{\rho}_{XX'AB} ||\sum_k V_X (\bra{k'}\ot
  \sqrt{Z_k}\ot \id_B)  \sg_{Z'AB} (\ket{k'}\ot \sqrt{Z_k}\ot \id_B)
  V_X\ad \big)\nonumber\\
  &\eqprop{\ref{a}}{\geq}&D_{K}\big(\tilde{\rho}_{XB} ||\sum_{j,k} \dya{j}\ot \Tr_{Z'A}((\dya{k'}\ot \sqrt{Z_k} X_j \sqrt{Z_k}\ot \id_B) \sg_{Z'AB})\big) \nonumber\\
  &\eqprop{\eqref{eqn9}}{\geq}&D_{K}(\tilde{\rho}_{XB} || c(X,Z) \id_X \ot \Tr_{Z'A}(\sg_{Z'AB})) \nonumber\\
  &\eqprop{\ref{c}}{=}&\log\frac{1}{c(X,Z)}+ D_{K}(\tilde{\rho}_{XB} ||  \id_X \ot \sg_{B}), \label{big_eqn}
\end{eqnarray}
\end{widetext}
where $\tilde{\rho}_{XX'AB} := V_X \rho_{AB}V_X\ad$.  Our use
of~\ref{a} in the first line involved the TPCPM $\rho\mapsto \Pi \rho
\Pi +(\id-\Pi)\rho (\id-\Pi)$ with $\Pi= V_Z V_Z \ad $.

If $H_K(A|B)$ takes the form $-D(\rho_{AB}||\id\ot\rho_B)$, then we
simply take $\sg_{Z'AB}=\rho_{Z'AB}$ in the above analysis to obtain
the uncertainty relation
$$H_K(X|B)+H_{\widehat K}(Z|C) \geq\log\frac{1}{c(X,Z)}.$$
Alternatively, if $H_K(A|B)$ takes the form
$\max_{\sg_B}[-D_K(\rho_{AB}||\id\ot\sg_B)]$, then we take $\sg_{Z'AB}$ to
be the state achieving the optimization for $H_K(Z|Z'AB)$ so that,
using~\eqref{big_eqn} and $H_K(X|B)\geq-D_K(\tilde{\rho}_{XB} || \id_X
\ot \sg_{B})$, we obtain
$$H_K(X|B)+H_{\widehat K}(Z|C) \geq\log\frac{1}{c(X,Z)}.$$

As presented so far, these relations hold for pure $\rho_{ABC}$.  More
generally, for mixed $\rho_{ABC}$, we can introduce a purification,
$\rho_{ABCD}$ and use $H_K(X|B)\geq H_K(X|BD)$, which follows from Property~\ref{a}.
\end{proof}

\subsection{Proof of Lemma~3}

\begin{proof}
As in the main article, let $\XC$ be the TPCPM from $\HC_A$ to $\HC_X$ given by
\begin{equation}
\XC:\rho_A\mapsto\sum_j\dyad{j}{j}_X\Tr(X_j\rho_A),\notag
\end{equation}
where $\{\ket{j}\}$ form an orthonormal basis in $\HC_X$. Then we have
\begin{align}
&D_K(\rho_{AB}||\id\ot\sigma_B)\notag\\
&\eqprop{\ref{a}}{\geq} D_K((\XC\ot\IC)(\rho_{AB})||\XC(\id)\ot\sigma_B) \notag\\
&=D_K(\tilde{\rho}_{XB}||\sum_j\Tr(X_j)\dyad{j}{j}\ot\sigma_B) \notag\\
& \eqprop{\eqref{eqn9}}{\geq} D_K(\tilde{\rho}_{XB}||\max_j\Tr(X_j)\id\ot\sigma_B) \notag\\
& \eqprop{\ref{c}}{=} D_K(\tilde{\rho}_{XB}||\id\ot\sigma_B)+\log (1/c'(X)). \notag
\end{align}
In the case $H_K(A|B)=-D_K(\rho_{AB}||\id\ot\rho_B)$, the result
immediately follows.  In the alternative case, take $\sigma_B$ to be
the state that achieves the optimum in $H_K(A|B)$ to establish the result.
\end{proof}

\ca
\subsection{General decoherence relations}
The following result applies to dual entropy pairs $(H_K,H_{\widehat{K}})$, e.g., of the form $(H,H)$, $(H_{\al},H_{2-\al})$, and $(\widehat{H}_{\min}, \widehat{H}_0)$.
\begin{lemma}\label{thm6}
Let $X=\{X_{j}\}$ be an arbitrary POVM on $A$ and $H_K(A|B)$ be such that, for all $\rho_{AB}$, $H_K(A|B)=-D_K(\rho_{AB}||\id\ot\rho_B)$ where $D_K$ satisfies Properties~\ref{a} and \ref{b}.  It follows that
for all $\rho_{ABC}$
\begin{equation}
\label{eqn10}
 H_{\widehat{K}}(X|C)\geq D_K(\rho_{AB}||\sum_j X_{j}\rho_{AB}X_{j}),
\end{equation}
with equality if $\rho_{ABC}$ is pure and the $X_j$ are projectors.
\end{lemma}

\begin{proof}
  Define $V_X$ as in~\eqref{eqnten}. Suppose $\rho_{ABC}$ is a general
  state and let $D$ purify it, so that
  $\tilde{\rho}_{XX'ABCD}=V_X\rho_{ABCD}V_X\ad$ is pure. Then applying
  the duality to this state gives
\begin{align}
 &H_{\widehat K}(X|C)= -H_K(X|X'ABD)\eqprop{\ref{a}}{\geq} -H_K(X|X'AB)\notag\\
 &= D_K(\tilde{\rho}_{XX'AB} || \id_X \ot \tilde{\rho}_{X'AB} ) \notag \\
 &\!\!\!\eqprop{\ref{a},\ref{b}}{\geq} D_K(\tilde{\rho}_{XX'AB} || V_X V_X\ad (\id_X \ot \tilde{\rho}_{X'AB}) V_X V_X\ad) \notag \\
 &\eqprop{\eqref{eqn4956}}{=} D_K(\rho_{AB} || V_X\ad (\id_X \ot \tilde{\rho}_{X'AB}) V_X) \notag \\
 &= D_K(\rho_{AB} || \sum_j (\bra{j}\ot \sqrt{X_j})\tilde{\rho}_{X'AB}  (\ket{j}\ot \sqrt{X_j})) \notag \\
 &= D_K(\rho_{AB} || \sum_j X_j\rho_{AB}X_j), \notag
\end{align}
where $\tilde{\rho}_{X'AB}=\sum_k \dya{k}\ot
\sqrt{X_k}\rho_{AB}\sqrt{X_k}$.  For pure $\rho_{ABC}$, the first
inequality is not needed, and, for projectors $\{X_j\}$, the second
inequality becomes an equality by invoking~\ref{b}.
\end{proof}
\cb

\subsection{Generalizations}
Of the properties we use to prove uncertainty relations,~\ref{c} is
the only one which establishes anything quantitative about the
entropy.  However, there are other entropies that scale differently in
this sense.  These may still satisfy modified uncertainty relations,
as outlined in the following lemma.
\begin{lemma}
\label{thm9}
Let $X=\{X_{j}\}$ and $Z=\{Z_{k}\}$ be arbitrary POVMs on $A$, and
$H_K(A|B)$ be such that either
$H_K(A|B)=-D_K(\rho_{AB}||\id\ot\rho_B)$ or
$H_K(A|B)=\max_{\sigma_B}[-D_K(\rho_{AB}||\id\ot\sigma_B)]$, for all
$\rho_{AB}$, where $D_K$ satisfies
Properties~\ref{a},~\ref{b} and~\ref{d}, and
\begin{equation}\label{eqngfd}
D_K(S||cT)=f(c)D_K(S||T)+g(c)
\end{equation}
for positive operators $S$ and $T$ and constant $c>0$.  It follows
that for all $\rho_{ABC}$
\begin{eqnarray*}
H_K(X|B)\!+\! f(c(X,Z)) H_{\widehat K}(Z|C)\!\!&\geq&\!\! g(c(X,Z))\notag\\
H_K(X|B)\!\!&\geq&\!\! g(c(X))\notag\\
f(c'(X))H_K(X|B)\!\!&\geq&\!\! g(c'(X))+H_K(A|B) \notag
\end{eqnarray*}
\end{lemma}
\begin{proof}
  This follows by replacing all uses of~\ref{c} by uses
  of~\eqref{eqngfd} in the proofs of Theorem~1 and
  Lemmas~2 and~3.
\end{proof}
Note that, in order to interpret these inequalities as uncertainty relations in the
usual sense, the functions $f(c)$ and $g(c)$ ought to be positive for
$0<c\leq 1$. (It is clear that setting $f(c)=1$ and $g(c)=\log(1/c)$ recovers our main results.) Furthermore, the same relations hold even if we can only
establish an inequality, $D_K(S||cT)\geq f(c)D_K(S||T)+g(c)$.

As a concrete example, consider the Tsallis relative and conditional
entropies \cite{Rastegin2011}. For positive operators $S$ and $T$, and density operator
$\rho_{AB}$, these can be defined as follows:
\begin{align}
&D_{T, \al}(S||T):=\lim_{\xi\rightarrow 0}\frac{1}{\alpha-1}(\Tr(S^\al (T+\xi\id)^{1-\al})-1) \notag\\
&H_{T,\al}(A|B):= -D_{T,\alpha}(\rho_{AB} || \id\ot \rho_B ) \notag\\
& \widehat {H}_{T,\al}(A|B):= \max_{\sg_B}[-D_{T,\alpha}(\rho_{AB} || \id\ot \sg_B )] \notag
\end{align}
for $\alpha\in(0,1)\cup (1,2]$. These are analogous to the R\'enyi
entropies defined in the main text with $\log(x)$ replaced by
$(x-1)$. Furthermore, they are directly related via
\begin{equation}\label{eqn132}
\log\left((1-\al)H_{T,\al}(A|B)+1\right)=(1-\al)H_\al(A|B).
\end{equation}
The duality relation $H_{T,\al}(A|B)=-H_{T,2-\al}(A|C)$ follows from
the analogous relation for the R\'enyi entropies.

The Tsallis relative entropies inherit Properties~\ref{a},~\ref{b}
and~\ref{d} from the R\'enyi entropies, and also
satisfy~\eqref{eqngfd} with $f(c)=c^{1-\al}$ and
$g(c)=\frac{1-f(c)}{1-\al}$.  They hence satisfy the corresponding
uncertainty relations of Lemma~\ref{thm9}.
We also note that the uncertainty relations for
Tsallis entropies can be shown by rearrangement of those of the
R\'enyi entropies using~\eqref{eqn132}.

\end{document}